\documentclass{llncs}

\usepackage[table]{xcolor}
\usepackage{graphicx}
\usepackage{indentfirst,amssymb,amsmath,latexsym}

\usepackage{algorithm}
\usepackage{algpseudocode}

\usepackage{multicol}

\newcommand{\RETURN}{{\textbf{return} }}

\newcommand{\vars}{{\tt vars}}

\newcommand{\lldots}{,\ldots ,}
\newcommand{\Poly}{{\tt P}}
\newcommand{\A}{{\tt A}}

\def\QED{$\square$\medskip}

\newtheorem{rem}{Remark}

\begin{document}

\title{A Polynomial Time Delta-Decomposition Algorithm for Positive DNFs \thanks{This work was supported by the grant of Russian Foundation for Basic Research No. 17-51-45125 and by the Ministry of Science and Education of the Russian Federation under the 5-100 Excellence Program.}}

\author{Denis~Ponomaryov}

\institute{Ershov Institute of Informatics Systems, Novosibirsk State University\newline
\email{ponom@iis.nsk.su}}







\maketitle

\begin{abstract}
We consider the problem of decomposing a positive DNF into a conjunction of DNFs, which may share a (possibly empty) given set of variables $\Delta$. This problem has interesting connections with traditional applications of positive DNFs, e.g., in game theory, and with the broad topic of minimization of boolean functions. We show that the finest $\Delta$-decomposition components of a positive DNF can be computed in polynomial time and provide a decomposition algorithm based on factorization of multilinear boolean polynomials.
\end{abstract}

\section{Introduction}
The interest in decomposition of positive DNFs stems from computationally hard problems in game theory, reliability theory, the theory of hypergraphs and set systems. A survey of relevant literature can be found in \cite{Bioch-2010}. In the context of voting games, boolean variables are used to represent voters and the terms of a positive DNF correspond to the winning coalitions, i.e., the groups of voters, who, when simultaneously voting in favor of an issue, have the power to determine the outcome of the vote (i.e., in this case the DNF is evaluated as true). Dual to them are blocking coalitions, i.e., those that force the outcome of the vote to be negative, irrespective of the decisions made by the remaining voters. The problem to find blocking coalitions with a minimal number of voters is easily shown to be equivalent to the hitting set problem, which is NP-complete.  Decomposing a DNF into components allows for reducing this problem to inputs having fewer variables. 

In this paper, we consider decomposition of a positive DNF into a \emph{conjunction of DNFs} sharing a given (possibly empty) subset of variables $\Delta$, with the remaining subsets of variables being disjoint (in this case we say that a DNF is $\Delta$-decomposable). In particular, each of the components has fewer variables than the original formula. Besides dimensionality reduction, decomposition of this kind facilitates finding a more compact representation of a positive boolean function. For example, the following DNF can be represented as a conjunction of two formulas:
\begin{equation}\label{Eq:DisjointDecompExample}
xa \vee xb \vee ya \vee yb \equiv (x\vee y) \ (a\vee b)
\end{equation}
i.e., it is $\varnothing$-decomposable into the components $x\vee y$ and $a \vee b$. The following DNF is not $\varnothing$-decompo\-sable, but it is $\{d_1,d_2\}$-decomposable:
\begin{equation}\label{Eq:ToExplainAlgo}
 xad_1 \vee xbd_1d_2 \vee yad_1d_2 \vee ybd_2 \equiv (xd_1 \vee yd_2) \ (ad_1 \vee bd_2)
\end{equation}
which can be easily verified by converting the expression into DNF.

In other words, $\Delta$-decomposition allows for common $\Delta$-variables between the components and partitions the remaining variables of a formula. Decomposition into variable disjoint components (i.e., $\Delta$-decomposition for $\Delta=\varnothing$) is known as disjoint conjunctive decomposition, or simply as AND-decomposition. The notion of OR-decomposition is defined similarly. 

The minimization of positive DNFs via decomposition in the sense above is related to open questions, not sufficiently addressed in the previous literature. For example, there is a fundamental work by Brayton et al. on the multilevel synthesis \cite{Brayton-2010}, which provides minimization methods with heuristics working well for arbitrary boolean functions. However, this contribution leaves space for research on minimization for special classes of functions, where the problem is potentially simpler. This is evidenced by the research in \cite{Boros} and \cite{PhDGursky}, for example.

It has been observed that the quality of multilevel decomposition (i.e., alternating AND/OR--decomposition) of DNFs strongly depends on the kind of decomposition used at the topmost level. As a rule, OR--decomposition has a priority over AND-decomposition in applications, since it is computationally trivial (while AND-decomposition is considered to be hard and no specialized algorithms for DNFs are known). However, choosing AND-decomposition at the topmost level may provide a more compact representation of a boolean function. For example, application of ``AND--first'' strategy gives a representation of the following positive DNF
\begin{eqnarray*}
absu \vee absv \vee absw \vee abtu \vee abtv \vee abtw \vee abxy \vee abxz \vee \\
  acsu \vee acsv \vee acsw \vee actu \vee actv \vee actw \vee acxy \vee acxz \vee  \\
  desu \vee desv \vee desw \vee detu \vee detv \vee detw \vee dexy \vee dexz\quad
\end{eqnarray*}
in the form
%
$(ab \vee ac \vee de) \  (su \vee sv \vee sw \vee tu \vee tv \vee tw \vee xy \vee xz).$
%
Further, OR--decomposition of the second component gives 
$\ {su \vee sv \vee sw \vee tu \vee tv \vee tw}$ and  ${xy \vee xz}$
(the first component similarly OR-decomposes syntactically). Finally, AND-decomposition of the obtained formulas gives a representation
\[(a(b \vee c) \vee de)((s \vee t)(u \vee v \vee w) \vee x(y \vee z),
\]
which is a read-once formula of depth 4 having 13 occurrences of variables.

On the contrary, Espresso\footnote{a well--known heuristic optimizer based on the work of Brayton et al., which is often used as a reference tool 
for optimization of boolean functions}, which implements OR--decomposition at the topmost level, gives a longer expression:
\[
x (a (c  \vee  b)  \vee  de) (z  \vee  y)  \vee  (a (c  \vee  b)  \vee  de) (t  \vee  s) (w  \vee  v  \vee  u)
\]

which is a formula of depth 5 having 18 occurrences of variables and this formula is not read-once.

Bioch \cite{Bioch-2010} studied (variable disjoint) decompositions of positive boolean functions in the form $\varphi \equiv F(G(X_A), X_B)$, where $\{X_A, X_B\}$ is a partition of the variables of $\varphi$ and $F,G$ are some (positive) boolean functions. The set $X_A$ is called a modular set of $\varphi$ in this case. By taking $\varphi=x_1 \vee \ldots \vee x_n$, one can see that the number of modular sets of $\varphi$ is exponential in $n$ (since any subset of the variables is modular). Bioch showed that one can compute a tree in time polynomial in the size of an input positive DNF, which succinctly represents all its modular sets. Given a modular set $X_A$, the corresponding component $G(X_A)$ can be also computed in polynomial time. These important results leave the question open however, which modular sets one should choose, when trying to find a compact representation of a boolean formula. For example, the modular tree for the DNF $\varphi$ from equation (\ref{Eq:DisjointDecompExample}) consists of the singleton variable subsets (plus the set of all the variables of $\varphi$ being modular by definition), from which one can obtain representations of the form $\varphi \ \equiv \ x(a\vee b) \vee ya \vee yb$ (and similarly for $y$ and $a,b$ selected). On the other hand, $\varnothing$-decomposition of this formula gives the representation $\varphi\equiv (x\vee y)(a\vee b)$, which is more compact. 

It has been shown in \cite{EmelyanovPonomaryov-2015pcs} that computing $\varnothing$-decomposition of a boolean function is coNP-hard in general, but for a positive DNF it can be computed in time polynomial in the size of the input formula (given as a string). In fact, it has been proved that $\varnothing$-decomposition reduces to factorization of a multilinear boolean polynomial efficiently obtained from the input positive DNF. For the latter problem the authors have provided a polynomial time algorithm based on the computation of formal derivatives. In this paper, we generalize the results from \cite{EmelyanovPonomaryov-2015pcs}. First, we provide an algorithm, which computes the finest $\varnothing$-decomposition components of a positive DNF. Since $\varnothing$-decomposable functions are expected to be rare, we consider the more general notion of $\Delta$-decomposition. The problem of computing $\Delta$-decomposition (for an arbitrary given $\Delta$) can be reduced to $\varnothing$-decomposition: it suffices to test whether each of the (exponentially many) DNFs, obtained from the input one for (all possible) evaluations of $\Delta$-variables, is $\varnothing$-decomposable with the same variable partition. The reduction holds for arbitrary boolean functions, not necessarily positive ones. We show however that for positive DNFs, it suffices to make $\varnothing$-decomposition tests only for a polynomial number of (positive) DNFs obtained from the input one. As a result, we obtain a polynomial time algorithm, which computes the finest $\Delta$-decomposition components of a positive DNF for a subset of variables $\Delta$. 


\section{Preliminaries}\label{Sect:BooleanExpressions}

A boolean expression is a combination of constants $0, 1$ and boolean variables using conjunction, disjunction, and negation. 
A boolean expression is a \emph{DNF} if it is a disjunction of \emph{terms} (where each term is a conjunction of literals and constants). We assume there are no double negations of variables in boolean expressions, no double occurrences of the same term in a DNF, or of the same literal in a term of a DNF. A DNF is \emph{positive} if it does not contain negated variables and the constant $0$. For a set of variables $V$, a \emph{$V$-literal} is a literal with a variable from $V$. We use the notation $\vars(\varphi)$ for the set of variables of an expression $\varphi$. For boolean expressions $\varphi$ and $\psi$, we write $\varphi\equiv\psi$ if they are logically equivalent.

\begin{definition}[$\Delta$-decomposability]
Let $\varphi$ be a boolean expression and $\Delta\subseteq\vars(\varphi)$ a subset of variables. The expression $\varphi$ is \emph{$\Delta$-decomposable} if it is equivalent to the conjunction of boolean expressions $\psi_1\lldots\psi_n$, where $n\geqslant 2$, such that the following holds:
\begin{itemize}
\item[a.] $\bigcup_{i=1\lldots n}\vars(\psi_i)\subseteq\vars(\varphi)$;
\item[b.] $\vars(\psi_i)\cap\vars(\psi_j)\subseteq\Delta$, for all $1\leqslant i,j \leqslant n$, $i\neq j$;
\item[c.] $\vars(\psi_i)\setminus\Delta\neq\varnothing$, for $i=1\lldots n$.
\end{itemize}
The expressions $\psi_1\lldots\psi_n$ are called \emph{($\Delta$-)decomposition components} of $\varphi$. If $\Delta=\varnothing$ and the above holds then we call $\varphi$ \emph{decomposable}, for short.
\end{definition}

Clearly, $\Delta$-decompo\-sition components can be subject to a more fine-grained decomposition, wrt the same or different delta's. It immediately follows from this definition that a boolean expression, which contains at most one non-$\Delta$-variable, is not $\Delta$-decomposable. Observe that conditions $a,c$ in the definition are important: if any of them is omitted then every boolean expression turns out to be $\Delta$--decomposable, for any (proper) subset of variables $\Delta$.

\begin{definition}[Finest Variable Partition wrt $\Delta$]
Let $\varphi$ be a boolean expression, $\Delta\subseteq\vars(\varphi)$ a subset, and $\pi=\{V_1\lldots V_{|\pi|}\}$ a partition of $\vars(\varphi)\setminus\Delta$. The expression $\varphi$ is said to be \emph{$\Delta$-decomposable with partition} $\pi$ if it has $\Delta$-decomposition components $\psi_i$, for $i=1\lldots |\pi|$, such that $\vars(\psi_i)=V_i\cup\Delta$. 

The \emph{finest variable partition} of $\varphi$ (wrt $\Delta$) is $\{\vars(\varphi)\}$ if $\varphi$ is not $\Delta$-decomposable. Otherwise it is the partition $\pi$, which corresponds to the non-$\Delta$-decomposable components $\psi_i$ of $\varphi$. 
\end{definition}

It will be clear from the results of this paper that for any $\Delta\subseteq\vars(\varphi)$, the finest variable partition of a positive DNF $\varphi$ is unique. In the general case, (e.g., for arbitrary boolean expressions), this property follows from the result proved in \cite{Ponomaryov-2008} for a broad class of logical calculi including propositional logic.  
As will be shown below, once the finest variable partition of a DNF $\varphi$ is computed, the corresponding (non-decomposable) components of $\varphi$ are easily obtained.

\medskip

Throughout the text, we use the term \emph{assignment} as a synonym for a consistent set of literals. Given a set of variables $V=\{x_1\lldots x_n\}$, where $n\geqslant 1$, a \emph{$V$-assignment} is a set of literals $\{l_1\lldots l_n\}$, where $l_i$ is a literal over variable $x_i$, for $i=1\lldots n$. 
Let $\varphi$ be a DNF, $V\subseteq\vars(\varphi)$ a subset, and $X$ a $V$-assignment such that there is a term in $\varphi$, whose set of $V$-literals is contained in $X$. Then the \emph{substitution} of $\varphi$ with $X$, denoted by $\varphi[X]$, is a DNF defined as follows:
\begin{itemize}
\item if there is a term $t$ in $\varphi$ such that every literal from $t$ is contained in $X$ then $\varphi[X]=1$ (and $X$ is called  \emph{satisfying assignment} for $\varphi$, notation: $X\models\varphi$);
\item otherwise $\varphi[X]$ is the DNF obtained from $\varphi$ by removing terms, whose set of $V$-literals is not contained in $X$, and by removing $V$-literals in the remaining terms.
\end{itemize}

For example, for the positive DNF $\varphi$ from (\ref{Eq:ToExplainAlgo}) we have $\varphi[\{d_1, \neg d_2\}]= xa$. 


\noindent For DNFs $\varphi$ and $\psi$, we say that $\varphi$ \emph{implies} $\psi$ if for any assignment $X$, it holds $X\not\models\varphi$ or $X\models\psi$. 
For a set of variables $V$ and a DNF $\varphi$, a \emph{projection} of $\varphi$ \emph{onto} $V$, denoted as $\varphi|_V$, is the DNF obtained from $\varphi$ as follows. If there is a term in $\varphi$ which does not contain a variable from $V$ then $\varphi|_V=1$. Otherwise $\varphi|_V$ is the DNF obtained from $\varphi$ by removing literals with a variable not from $V$ in the terms of $\varphi$. It should be clear from this definition that $\varphi$ implies $\varphi|_V$ (in the literature, projection is also known as a \emph{uniform interpolant} or the \emph{strongest consequence} of $\varphi$ wrt $V$). 
For instance, for the DNF $\varphi$ from (\ref{Eq:ToExplainAlgo}) we have $\varphi |_{\{x,y,d_1,d_2\}}= xd_1\vee xd_1d_2 \vee yd_1d_2 \vee yd_2$. 

\begin{lemma}[Decomposition Components as Projections]\label{Lem:ComponentsAsProjections}
Let $\varphi$ be a DNF, which is $\Delta$-decomposable with a variable partition $\pi=\{V_1\lldots V_n\}$. Then $\varphi |_{U_1},$ $.. ,\varphi |_{U_n}$, where $U_i=V_i\cup\Delta$, $i=1\lldots n$, are $\Delta$-decomposition components of $\varphi$.
\end{lemma}
\begin{proof}
Since $\varphi$ implies $\varphi |_{U_i}$, for $i=1\lldots n$, it suffices to demonstrate that $\bigwedge_{i=1\lldots n}\varphi |_{U_i}$ implies $\varphi$. Assume $\varphi\equiv \varphi_1\wedge\ldots\wedge\varphi_n$, where $\varphi_1\lldots\varphi_n$ are $\Delta$-decomposition components of $\varphi$, with $\vars(\varphi_i)=U_i$, for $i=1 \lldots n$. We show that $\varphi |_{U_i}$ implies $\varphi_i$, for all $i=1 \lldots n$.
For suppose there is a $U_i$-assignment $X$ such that $X\models \varphi |_{U_i}$ and $X\not\models\varphi_i$, for some $i\in\{1\lldots n\}$. Then by the definition of the projection $\varphi |_{U_i}$ there exists an assignment $X'\supseteq X$ such that $X'\models\varphi$ and $X'\not\models\varphi_i$, which is a contradiction, since $\varphi$ implies $\varphi_i$, for $i=1\lldots n$. As $\varphi_1\wedge\ldots\wedge\varphi_n$ implies $\varphi$, we conclude that $\bigwedge_{i=1\lldots n}\varphi |_{U_i}$ implies $\varphi$.  \QED
\end{proof}


Let $V$ be a set of variables, $\Delta\subseteq V$ a subset, and $\A$ a set of $V$-assignments. By $\A|_\Delta$ we denote the set of all $\Delta$-assignments $d$, for which there is an assignment $X\in\A$ such that $d\subseteq X$. For a $\Delta$-assignment $d$, the notation $\A\langle d\rangle$ stands for the set of $(V\setminus\Delta)$-assignments $X$ such that $X\cup d\in\A$. 
Let $V_1, V_2$ be disjoint sets of variables and for $i=1,2$, let $\A_i$ be a set of $V_i$-assignments. Then the notation $\A_1\bowtie\A_2$ stands for the set of all assignments $X_1\cup X_2$ such that $X_1\in\A_1$ and $X_2\in\A_2$. 
The intuitive relationship between the cartesian combinations of assignments is illustrated by the following remark and is put formally in the subsequent Lemma \ref{Lem:DecompositionCriterion}. 

\begin{rem}[Conjunction of DNFs is Similar to Cartesian Product]\label{Re_ConjunctionOfDNFs} \text{}\newline
Taking the conjunction of DNFs $\xi_1\vee\ldots \vee\xi_m$ and $\zeta_1\vee \ldots \vee\zeta_n$ gives a DNF, which has the form $\bigvee(\xi_i\wedge\zeta_j),$ for all pairs $i,j$, with $1\leqslant i \leqslant m$, $1\leqslant j \leqslant n$.
\end{rem}

\begin{lemma}[Decomposability Criterion]\label{Lem:DecompositionCriterion}
Let $\varphi$ be a DNF and $\A$ the set of satisfying assignments for $\varphi$. Then $\varphi$ is $\Delta$-decomposable with a partition $\pi=\{V_1\lldots V_{|\pi|}\}$ iff for all $d\in\A|_\Delta$ it holds that $\A\langle d \rangle = \A\langle d \rangle|_{V_1}\bowtie\ldots\bowtie \A\langle d \rangle|_{V_{|\pi|}}$.
\end{lemma}

In this paper, we are concerned with the problem of finding the finest variable partition of a positive DNF $\varphi$ wrt a subset $\Delta$ of its variables. By Lemma \ref{Lem:ComponentsAsProjections}, decomposition components of $\varphi$ can be easily obtained from the finest variable partition. We assume that boolean expressions are given as strings and thus, the size of an expression $\varphi$ is the length of the string, which represents $\varphi$.

For the sake of completeness, first we describe a factorization algorithm for multilinear boolean polynomials, which is based on the results from \cite{EmelyanovPonomaryov-2015pcs,EmelyanovPonomaryov-2014psi}. Then we provide a $\varnothing$-decomposition algorithm for a positive DNF $\varphi$ based on factorization of a boolean polynomial, which is obtained from $\varphi$ by a simple syntactic transformation. Finally, we demonstrate that $\Delta$-decomposition of a positive DNF reduces to $\varnothing$-decomposition of (a polynomial number of) positive DNFs obtained from the input one and devise the corresponding polynomial time $\Delta$-decomposition algorithm.

\section{Factorization of Boolean Polynomials}\label{Sect:Polynomials}

In \cite{ShpilkaVolkovich-2010icalp}, Shpilka and Volkovich established the prominent result on the equivalence of polynomial factorization and identity testing. It follows from their result that a multilinear boolean polynomial
can be factored in time cubic in the size of the polynomial given as a string. This result has been rediscovered in \cite{EmelyanovPonomaryov-2015pcs,EmelyanovPonomaryov-2014psi}, where the authors have provided a factorization algorithm based on the computation of derivatives of multilinear boolean polynomials, which allows for deeper optimizations. Without going into implementation details, we employ this result here to formulate an algorithm, which computes the finest $\varnothing$-decomposition components of a positive DNF $\varphi$.
Hereafter we assume that polynomials do not contain double occurrences of the same monomial.

\begin{definition}
A polynomial $F$ is \emph{factorable} if $F=G_1\cdot \ldots \cdot G_n$, where $n\geqslant 2$ and $G_1, \ldots, G_n$ are some non-constant polynomials. Otherwise $F$ is \emph{irreducible}. The polynomials $G_1, \ldots, G_n$ are called \emph{factors} of $F$.
For a polynomial $F$, the \emph{finest variable partition} of $F$ is $\{\vars(F)\}$ if $F$ is irreducible and otherwise consists of the sets of the variables of the irreducible factors of $F$.
\end{definition}

It is important to stress that we consider here multilinear polynomials (every variable can occur only in the power of $\leqslant 1$) and thus, the factors are polynomials \emph{over disjoint sets of variables}. 
Note that the finest variable partition of a multilinear boolean polynomial is unique, since the ring of these polynomials is a unique factorization domain.
We now formulate the first important observation, which is a strengthening of Theorem 5 from \cite{EmelyanovPonomaryov-2015pcs}.

\begin{theorem}[Computing Finest Variable Partition for Polynomial]\label{Teo:Factorization}
For a multilinear boolean polynomial $F$, the (unique) finest partition of the variables of $F$ can be found in time polynomial in the size of $F$. 
\end{theorem}

It is proved in \cite{EmelyanovPonomaryov-2015pcs} that testing whether $F$ is factorable and computing its factors can be done in time polynomial in the size of $F$ given as a string. By applying the factorization procedure to the obtained factors recursively, one obtains a partition of the variables of $F$, which corresponds to the irreducible factors of $F$. This is implemented in {\tt FindPartition} procedure given below, which is a modification of the factorization algorithm from \cite{EmelyanovPonomaryov-2015pcs}. It is also shown in \cite{EmelyanovPonomaryov-2015pcs,EmelyanovPonomaryov-2014psi} that once a partition of variables, which corresponds to the factors of $F$ is computed, the factors can be easily obtained as projections of $F$ onto the components of the partition (see the notion of projection below).

\smallskip

The {\tt FindPartition} procedure takes a boolean polynomial $F$ as an input and outputs the finest partition of $\vars(F)$ in time polynomial in the size of $F$. A few notations are required. For a polynomial $F$, we denote by $\vars(F)$ the set of the variables of $F$. For a variable $x\in \vars(F)$ and a value $a\in\{0,1\}$, we denote by $F_{x=a}$ the polynomial obtained from $F$ by substituting $x$ with $a$. 
Given a set of variables $V$ and a monomial $m$, the \emph{projection} of $m$ onto $V$ (denoted as $m|_V$) is $1$ if $m$ does not contain any variable from $V$, or is equal to the monomial obtained from $m$ by removing all the variables not contained in $V$, otherwise. The \emph{projection} of a polynomial $F$ onto $V$, denoted as $F|_V$, is the polynomial obtained by projecting the monomials of $F$ onto $V$ and by removing duplicate monomials. 

\smallskip

Lines \ref{AlgoLine:CheckConstantPolynomial}-\ref{AlgoLine:ReturnNull4ConstantPolynomial} of {\tt FindPartition} is a test for a simple sufficient condition for irreducibility: if a polynomial is a constant then it cannot be factorable. Lines \ref{AlgoLine:TrivialFactorCheckStart}-\ref{AlgoLine:TrivialFactorCheckEnd} implement a test for trivial factors: if some variable $z$ is present in every monomial of $F$, then $z$ is an irreducible factor. In the recursive part of the procedure, the remaining sets from the finest variable partition of $F$ are computed as the values of the variable $\Sigma$ and are added to {\tt FinestPartition}.


\begin{multicols}{2}

{\small
\begin{algorithmic}[1]
\Procedure{FindPartition}{$F$}

  \If{$F==0$ or $F==1$} \label{AlgoLine:CheckConstantPolynomial}
   \State \RETURN $\vars(F)$   
  \EndIf \label{AlgoLine:ReturnNull4ConstantPolynomial}
   \For{$z$ a variable occurring in every monomial of $F$}\label{AlgoLine:TrivialFactorCheckStart}
      \State FinestPartition.add($\{z\}$)
      \State $F \gets F_{z=1}$
   \EndFor
   \If{$F$ does not contain any variables}
     \State \RETURN FinestPartition
   \EndIf	 	      
   \If{$F$ contains a single variable, e.g., $x$}
     \State FinestPartition.add($\{x\}$)
     \State \RETURN FinestPartition
   \EndIf	 \label{AlgoLine:TrivialFactorCheckEnd}	


 \State $V \gets$ variables of $F$  

\smallskip

  \Repeat 
    \State $\Sigma \gets \varnothing$;  $\ \ \ F \gets F|_V$ \label{AlgoLine:Iteration}
    \State pick a variable $x$ from $V$ \label{AlgoLine:XVarChoice}
    \State $\Sigma$.add($x$); $\ \ \ V \gets V\setminus\{x\}$ 
    
    \State $G\gets F_{x=0} \cdot \frac{\partial{F}}{\partial{x}}$
 
 \smallskip
 
    \For{a variable $y$ from $V$}\label{AlgoLine:YVarChoice}
			 \If{$\frac{\partial{G}}{\partial{y}}\neq 0$}
				\State $\Sigma$.add($y$)
			 \EndIf     
     \EndFor
   \State FinestPartition.add($\Sigma$)  \label{AlgoLine:Add2Partition} 
   \State $V \gets V\setminus\Sigma$
  
  \Until $V=\varnothing$
\smallskip  
\State \RETURN FinestPartition
\smallskip
\EndProcedure
\end{algorithmic}
}

\end{multicols}

\section{$\varnothing$-decomposition of Positive DNFs}\label{Sect:0Decomp} 

A term $t$ of a DNF $\varphi$ is called \emph{redundant} in $\varphi$ if there exists another term $t'$ of $\varphi$ such that every literal of $t'$ is present in $t$ (i.e., $t'\subseteq t$). For example, the term $xy$ is redundant in $xy \vee x$. It is easy to see that removing redundant terms gives a logically equivalent DNF. 

Let us note the following simple fact: 

\begin{lemma}[Existence of Positive Components]\label{Lem:PositiveDecomposesIntoPositive}
Let $\varphi$ be a positive boolean expression and $\Delta\subseteq\vars(\varphi)$ a subset of variables. If $\varphi$ is $\Delta$-decomposable then it has decomposition components, which are positive expressions. 
\end{lemma}
\begin{proof}
It is known (e.g., see Theorem 1.21 in \cite{CramaHammer-2011}) that a boolean expression $\psi$ is equivalent to a positive one in a variable $x$ iff for the set of satisfying assignments $\A$ for $\varphi$ the following property holds: if $\{l_1\lldots l_n,\neg x\}\in\A$, where $l_1\lldots l_n$ are literals, then $\{l_1\lldots l_n, x\}\in\A$. Clearly, this property is preserved under decomposition: if a set of assignments $\A$ satisfies the property and it holds that $\A=\A_1\bowtie\ldots\bowtie \A_n$, then so do the sets $\A_i$, for $i=1\lldots n$. Thus, the claim follows directly from Lemma \ref{Lem:DecompositionCriterion}.  \QED
\end{proof}

The next important observation is a strengthening of the result from \cite{EmelyanovPonomaryov-2015pcs}, which established the complexity of $\varnothing$-decomposition for positive DNFs.

\begin{theorem}[Computing the Finest Variable Partition wrt $\Delta=\varnothing$] \label{Thm:DisjointDecompositionPosDNF}
The finest variable partition of a positive DNF $\varphi$ can be computed in time polynomial in the size of $\varphi$.
\end{theorem}

Let $\Poly$ be a 1-1 mapping, which for a positive DNF $\varphi$ gives a multilinear boolean polynomial $\Poly(\varphi)$ over $\vars(\varphi)$ obtained by replacing the conjunction and disjunction with $\cdot$ and $+$, respectively. 
The theorem is proved by showing that decomposition components of a positive DNF $\varphi$ can be recovered from factors of a polynomial $\Poly(\psi)$ constructed for a DNF $\psi$, which is obtained from $\varphi$ by removing redundant terms. The idea is illustrated in {\tt $\varnothing$Decompose}  procedure below, which for a given positive DNF $\varphi$ computes the finest variable partition of $\varphi$. It relies on the factorization procedure from Section \ref{Sect:Polynomials} and is employed as a subroutine in $\Delta$-decomposition algorithm in Section \ref{Sect:DeltaDecomp}. The procedure uses a simple preprocessing, which removes redundant terms. The preprocessing also allows for detecting those variables (line \ref{AlgoLine:InessentialVarsStart} of the procedure) that $\varphi$ does not depend on. By the definition of decomposability, these variables are decomposition components of $\varphi$, so they are added as singleton sets into the resulting finest variable partition (at line \ref{AlgoLine:InessentialVarsEnd}).

\begin{multicols}{2}

{\small
\begin{algorithmic}[1]
\Procedure{$\varnothing$Decompose}{$\varphi$}
 
 \State FinestPartition $\gets \varnothing$ 
 \State $\psi \gets$ \Call{RemoveRedundTerms}{$\varphi$}
 
  
  \State FinestPartition $\gets $ \newline \Call{FindPartition}{$\Poly(\psi)$} \Comment see Sect. \ref{Sect:Polynomials}
  
\medskip  

\ForAll{$x\in\vars(\varphi)\setminus\vars(\psi)$} \label{AlgoLine:InessentialVarsStart}
   \State FinestPartition.add($\{x\}$) \label{AlgoLine:InessentialVarsEnd}
 \EndFor 

 \State \RETURN FinestPartition

\EndProcedure
\end{algorithmic}

\columnbreak

\begin{algorithmic}[1]
\Procedure{RemoveRedundTerms}{$\varphi$}
  \ForAll{terms $t$ in $\varphi$}
    \If{there exists a term $t'$ in $\varphi$ s.t. $t'\subseteq t$}
      \State remove $t$ from $\varphi$
    \EndIf 
 \EndFor
 \State \RETURN $\varphi$
\EndProcedure
\end{algorithmic}
}

\end{multicols}

\section{$\Delta$-decomposition of Positive DNFs}\label{Sect:DeltaDecomp}

\begin{definition}[$\Delta$-atom]
For a positive DNF $\varphi$ and a subset $\Delta\subseteq\vars(\varphi)$, the set of $\Delta$-variables of a term of $\varphi$ is called \emph{$\Delta$-atom} of $\varphi$. 
\end{definition}

Note that by definition a $\Delta$-atom can also be the empty set. Let $U$ be the set of unions of $\Delta$-atoms of $\varphi$. Given a set $X\in U$, we introduce the notation $\varphi\langle X\rangle$ as a shortcut for the DNF $\varphi[X\cup\bar{X}]$, where $\bar{X}=\{\neg x \mid x\in\Delta\setminus X\}$.

\medskip 

Let $\pi$ be a partition of $\vars(\varphi)\setminus\Delta$.  
We say that a boolean expression $\psi$ \emph{supports} $\pi$ if every set from the finest variable partition of $\psi$ wrt $\Delta$ is contained in some set from $\pi$. It is easy to see that if $\varphi$ is $\Delta$-decomposable with $\pi$, then $\varphi[X]$ supports $\pi$, for any set of literals $X$ such that $\varphi[X]$ is defined.

\medskip

We formulate two lemmas that are the key to the main result, Theorem \ref{Thm:DeltaDecompositionPosDNF}, in this section. 

\begin{lemma}[$\Delta$-Decomposability Criterion for Positive DNF]\label{Lem:DecompCriterionPosDNF}
Let $\varphi$ be a positive DNF, $\Delta\subseteq \vars(\varphi)$ a subset, and $U$ the set of unions of $\Delta$-atoms of $\varphi$. Then $\varphi$ is $\Delta$-decomposable with a variable partition $\pi$ iff $\varphi\langle X\rangle$ supports $\pi$, for all $X\in U$.
\end{lemma}
\begin{proof}
$(\Rightarrow)$: Take $X\in U$. Since $\varphi$ is positive, $X$ is a consistent set of literals, $\varphi\langle X\rangle$ is defined, and clearly,  supports $\pi$.

\medskip

$(\Leftarrow)$: Let $\pi=\{V_1\lldots V_{|\pi|}\}$, $\A$ be the set of satisfying assignments for $\varphi$, and $d\in \A|_\Delta$ a $\Delta$-assignment. Then there is $X\in U$ such that $X\subseteq d$, since $\varphi$ is a DNF. Let $X$ be the maximal set from $U$ with this property. Then we have $\varphi[d]=\varphi\langle X\rangle$, so $\varphi[d]$ supports $\pi$. This yields $\A\langle d \rangle = \A\langle d \rangle|_{V_1}\bowtie\ldots\bowtie \A\langle d \rangle|_{V_{|\pi|}}$   and since $d$ was arbitrarily chosen, it follows from Lemma \ref{Lem:DecompositionCriterion} that $\varphi$ is $\Delta$-decomposable with $\pi$.  \QED
\end{proof}

\begin{lemma}[Decomposition Lemma]\label{Lem:DecompositionLemma}
Let $\varphi_1,\ldots ,\varphi_n$, where $n\geqslant 1$, be DNFs with the following property: for all $1\leqslant j,k \leqslant n$ there is a subset $I\subseteq\{1,\ldots ,n\}$, with $j,k\in I$, such that $\bigvee_{i\in I}\varphi_i$ is decomposable with $\pi$. Then so is $\varphi_1\vee\ldots\vee\varphi_n$.
\end{lemma}
\begin{proof}
Let $\pi=\{X,Y\}$ and denote $\varphi=\varphi_1\vee\ldots \vee\varphi_n$. By Lemma \ref{Lem:ComponentsAsProjections} we need to show that $\varphi \equiv \varphi|_X \wedge \varphi|_Y$, which is equivalent to:
\begin{align}\label{Eq:decomposition}
\varphi \equiv (\varphi_1|_X\vee\ldots\vee\varphi_n|_X) \wedge  (\varphi_1|_Y\vee\ldots\vee\varphi_n|_Y)
\end{align}
Observe that the right-hand side of this equation can be written as the expression $D=\bigvee_{1\leqslant j,k \leqslant n}\varphi_j|_X\ \varphi_k|_Y$. Take any $j,k\in\{1,\ldots ,n\}$. By the condition of the lemma there is a subset $I\subseteq\{1,\ldots ,n\}$, with $j,k\in I$, such that $\bigvee_{a,b\in I}\ \varphi_a|_X \ \varphi_b|_Y\ \equiv \ \bigvee_{i\in I}\varphi_i$. That is, a disjunction of formulas from $D$ containing both, $\varphi_j|_X \ \varphi_k|_Y$ and $\varphi_k|_X \ \varphi_j|_Y$ is a equivalent to a disjunction of formulas from $\varphi$. Since the choice of $j,k$ was arbitrary, we conclude that (\ref{Eq:decomposition}) holds and thus, the lemma is proved. \QED
\end{proof}

\begin{theorem}[Computing the Finest Variable Partition wrt $\Delta$]\label{Thm:DeltaDecompositionPosDNF}
Given a positive DNF $\varphi$ and a subset $\Delta\subseteq\vars(\varphi)$, the finest variable partition of $\varphi$ wrt $\Delta$ can be computed in time polynomial in the size of $\varphi$.
\end{theorem}
\begin{proof}
Let $A$ be the set of $\Delta$-atoms of $\varphi$ and $U$ consist of all unions of sets from $A$. Note that $|A|$ is bounded by the size of $\varphi$, while $|U|$ is exponential. By Lemma \ref{Lem:DecompCriterionPosDNF}, $\varphi$ is $\Delta$-decomposable with a partition $\pi$ iff  $\varphi\langle X\rangle$ supports $\pi$, for all $X\in U$. 

\medskip

For any $X\in U$, we have $\vars(\varphi\langle X\rangle)\subseteq\vars(\varphi)$.  Observe that $\varphi\langle X\rangle$ is equivalent to the DNF $\psi=\varphi\langle X\rangle\vee t$, where $t$ is a term redundant in $\psi$, $\vars(t)=\vars(\varphi)\setminus\vars(\varphi\langle X\rangle)$ (in case $\vars(t)=\varnothing$ we assume that $\psi=\varphi\langle X\rangle$) and it holds $\vars(\psi)=\vars(\varphi)$. Therefore, $\varphi\langle X\rangle$ supports $\pi$ iff $\psi$ is decomposable with $\pi$. 

\medskip

For any $X\in U$, $\varphi\langle X\rangle$ is equivalent to $\varphi_1\vee\ldots\vee\varphi_n$ where $n\geqslant 1$ and for $i=1\lldots n$, $\ \varphi_i=\varphi\langle a_i\rangle $, where $a_i\subseteq X$, a $\Delta$-atom of $\varphi$. Notice further that $\varphi\langle X\rangle$ is equivalent to $\varphi'_1\vee\ldots\vee\varphi'_n$, where $\varphi'_i=\varphi_i \vee t_i$ and $t_i$ is a redundant term as introduced above.  

\medskip

By Lemma \ref{Lem:DecompCriterionPosDNF}, if $\varphi$ is $\Delta$-decomposable with a partition $\pi$ then $\varphi\langle a_1\cup a_2\rangle$ supports $\pi$, for any $a_1,a_2\in A$. For the other direction, if $\varphi\langle a_1\cup a_2\rangle$ supports $\pi$, for any $a_1,a_2\in A$ then the condition of Lemma \ref{Lem:DecompositionLemma} holds for $\varphi'_1\vee\ldots\vee\varphi'_n$. It follows that $\varphi\langle X\rangle$ supports $\pi$, for any $X\in U$ and hence by Lemma \ref{Lem:DecompCriterionPosDNF}, $\varphi$ is $\Delta$-decomposable with $\pi$.

\medskip

By Theorem \ref{Thm:DisjointDecompositionPosDNF}, a variable partition $\sigma$, which corresponds to the finest decomposition of $\varphi\langle a_1\cup a_2\rangle$, can be found in time polynomial in the size of $\varphi\langle a_1\cup a_2\rangle$ (and hence, in the size of $\varphi$, as well). For any  variables $x,y\in\vars(\varphi)$ and a set $S\in\sigma$, if $x,y\in S$ then $x$ and $y$ cannot belong to different decomposition components of $\varphi\langle a_1\cup a_2\rangle$. 

\medskip
Let $\sim$ be an equivalence relation on $\vars(\varphi)$ such that $x\sim y$ iff there are $a_1,a_2\in A$ such that $x$ and $y$ belong to the same component of the finest variable partition of $\varphi\langle a_1\cup a_2\rangle$. Since $|A|$ is bounded by the size of $\varphi$, one can readily verify that the equivalence classes wrt $\sim$ can be computed in time polynomial in the size of $\varphi$ and are equal to its finest variable partition. \QED 
\end{proof}

We conclude the paper with a description of {\tt $\Delta$Decompose} procedure, which for a positive DNF $\varphi$ and a (possibly empty) subset $\Delta\subseteq\vars(\varphi)$ computes the finest variable partition of $\varphi$ wrt $\Delta$ and outputs $\Delta$-decomposition components, which correspond to the partition. 

\medskip

In Lines \ref{AlgoLine_DeltaAtomsStart}-\ref{AlgoLine_DeltaAtomsEnd} of the procedure, a set of $\Delta$-atoms of $\varphi$ is computed, while skipping those ones, which subsume some term of $\varphi$. Clearly, if there is a term $t$ of $\varphi$, which consists only of $d$-variables for some subset $d\subseteq\Delta$, then it holds $\varphi[d]=1$, which implies that $\varphi\langle d\rangle$ supports any partition $\pi$ of $\vars(\varphi)\setminus\Delta$ (at this point $\varphi$ necessarily contains at least $2$ non-$\Delta$-variables due to the test in line \ref{AlgoLine_CheckMore2NonDeltaVars}). Therefore, these atoms are irrelevant in computing decomposition and they can be omitted (similarly, the unions of $\Delta$-atoms in line \ref{AlgoLine_CheckDoesntCoverATerm}). 

\smallskip

Lines \ref{AlgoLine_TestPairsDeltaLitSetStart}-\ref{AlgoLine_TestPairsDeltaLitSetEnd} implement a call for computing the finest variable partition wrt the empty $\Delta$ for each DNF $\varphi\langle L\rangle$ obtained from $\varphi$ for a union $L$ of relevant $\Delta$-atoms. The result is a family of partitions, which are further aligned by computing equivalence classes on the variables of $\varphi$. This is implemented in {\tt AlignPartitions} procedure by computing connected components of a graph, in which vertices correspond to the variables of $\varphi$. 

\smallskip

Finally, in lines \ref{AlgoLine_OutputComponentsStart}-\ref{AlgoLine_OutputComponentsEnd} the decomposition components of $\varphi$ are computed as projections onto the sets of variables corresponding to the finest partition. The components are cleaned up by removing redundant terms and are sent to the output.


\begin{multicols}{2}

{\small
\begin{algorithmic}[1]
\Procedure{$\Delta$Decompose}{$\varphi$, $\Delta$}
 \State FinestPartition $\gets \varnothing$
 \State Components $\gets \varnothing$
 

\smallskip
\If{$\varphi$ contains at most one non-$\Delta$-variable} \label{AlgoLine_CheckMore2NonDeltaVars}
  \State \RETURN $\{\varphi\}$  \Comment{$\varphi$ is not $\Delta$-decomposable}
\EndIf

\medskip



 \State $\Delta$Atoms $\gets \varnothing$ 
 
 \For{every term $t$ of $\varphi$, which contains at least one non-$\Delta$-variable} \label{AlgoLine_DeltaAtomsStart}
 \vspace{-0.3cm}
    \State \hspace{-0.5cm} $\Delta$Atoms.add(the set of $\Delta$-variables of $t$)
 \EndFor \label{AlgoLine_DeltaAtomsEnd}

  
 \ForAll{$a_1, a_2$ from $\Delta$Atoms} \label{AlgoLine_TestPairsDeltaLitSetStart}
 \State     \hspace{-0.5cm} $L \gets a_1\cup a_2$
\If{there is no term $t$ in $\varphi$, whose every variable is from $L$} \label{AlgoLine_CheckDoesntCoverATerm}
     \State \hspace{-1.2cm} PartitionForL $\gets$ \Call{$\varnothing$Decompose}{$\varphi\langle L\rangle$} \Comment{see Sect. \ref{Sect:0Decomp}}
     \State \hspace{-0.9cm} PartitionFamily.add(PartitionForL)    	    	
    \EndIf 
 \EndFor  \label{AlgoLine_TestPairsDeltaLitSetEnd}
  
\State FinestPartition $\gets$ \newline \phantom{aaa} \Call{AlignPartitions}{PartitionFamily} \label{AlgoLine_CallAlignPartitions}

\If{FinestPartition.isSingleton}
           \State \RETURN $\{\varphi\}$  \Comment{$\varphi$ is not $\Delta$-decomposable}
           \columnbreak
\Else 

\For{$V \in$ FinestPartition} \label{AlgoLine_OutputComponentsStart}
	   \State \hspace{-1.3cm} $\psi\gets$\Call{RemoveRedundTerms}{$\varphi|_{V\cup\Delta}$}  \Comment{see Sect. \ref{Sect:0Decomp}}
       \State Components.add($\psi$)    
\EndFor \label{AlgoLine_OutputComponentsEnd}

 \State \RETURN Components
\EndIf
\vspace{0.2cm}
\EndProcedure
\end{algorithmic}


\bigskip

\begin{algorithmic}[1]
\Procedure{AlignPartitions}{PFamily}
 \State G$\gets\varnothing$ \Comment{a graph with vertices being vars. of $\varphi$} 
 
 \For{Partition $\in$ PFamily}
   \For{VarSet $\in$ Partition}
     \State G.add(a path involving all $x \in$ VarSet)   
    \EndFor 
  \EndFor  
  
\medskip  
  
  \State ResultPartition $\gets \varnothing$  

  \For{C a connected component of $G$}
    \State ResultPartition.add(the set of vars from C)
  \EndFor  
  
 \State \RETURN ResultPartition 
\EndProcedure
\end{algorithmic}
\medskip
} 

\end{multicols}

\bibliographystyle{splncs}
\bibliography{references}

\begin{thebibliography}{1}

\bibitem{Bioch-2010}
Bioch, J.C.:
\newblock Decomposition of {Boolean} functions.
\newblock In Crama, Y., Hammer, P.L., eds.: Boolean Models and Methods in
  Mathematics, Computer Science, and Engineering. Volume 134 of {Encyclopedia
  of Mathematics and its Applications}.
\newblock Cambridge University Press, New York, NY, USA (2010)  39--78

\bibitem{Brayton-2010}
Villa, T., K.~Brayton, R., Sangiovanni-Vincentelli, A.:
\newblock Synthesis of multilevel boolean networks.
\newblock In Crama, Y., Hammer, P.L., eds.: Boolean Models and Methods in
  Mathematics, Computer Science, and Engineering. Volume 134 of {Encyclopedia
  of Mathematics and its Applications}.
\newblock Cambridge University Press, New York, NY, USA (2010)  675--722

\bibitem{Boros}
Boros, E.:
\newblock Horn functions.
\newblock In Crama, Y., Hammer, P.L., eds.: Boolean Functions: Theory,
  Algorithms, and Applications. Volume 134 of {Encyclopedia of Mathematics and
  its Applications}.
\newblock Cambridge University Press, New York, NY, USA (2011)  269--325

\bibitem{PhDGursky}
Gursky, S.:
\newblock Special classes of boolean functions with respect to the complexity
  of their minimization.
\newblock PhD Thesis, Charles University in Prague (2014)

\bibitem{EmelyanovPonomaryov-2015pcs}
Emelyanov, P., Ponomaryov, D.:
\newblock Algorithmic issues of conjunctive decomposition of boolean formulas.
\newblock Programming and Computer Software \textbf{41}(3) (2015)  162--169

\bibitem{Ponomaryov-2008}
Ponomaryov, D.:
\newblock On decomposability in logical calculi.
\newblock Bulletin of the Novosibirsk Computing Center \textbf{28} (2008)
  111--120

\bibitem{EmelyanovPonomaryov-2014psi}
Emelyanov, P., Ponomaryov, D.:
\newblock On tractability of disjoint {AND}--decomposition of boolean formulas.
\newblock In: Proceedings of the PSI 2014: 9$^{th}$ Ershov Informatics
  Conference. Volume 8974 of Lecture Notes in Computer Science., Springer
  (2015)  92--101

\bibitem{ShpilkaVolkovich-2010icalp}
Shpilka, A., Volkovich, I.:
\newblock On the relation between polynomial identity testing and finding
  variable disjoint factors.
\newblock In: Proceedings of the 37$^{th}$ International Colloquium on
  Automata, Languages and Programming. Part 1 (ICALP 2010). Volume 6198 of
  Lecture Notes in Computer Science., Springer (2010)  408--419

\bibitem{CramaHammer-2011}
Crama, Y., Hammer, P.L.:
\newblock Boolean Functions - Theory, Algorithms, and Applications. Volume 142
  of Encyclopedia of mathematics and its applications.
\newblock Cambridge University Press (2011)

\end{thebibliography}

\end{document}